\documentclass[letter, 11pt,onecolumn, conference]{ieeeconf}      

\IEEEoverridecommandlockouts                              

\usepackage{amsfonts}
\usepackage{amsmath}
\usepackage{amssymb}
\usepackage{color}
\usepackage{graphicx}
\usepackage{float}
\usepackage[mathscr]{eucal}
\usepackage{amsfonts}
\usepackage{float}
\usepackage[colorlinks,linkcolor=blue]{hyperref}
\usepackage[latin1]{inputenc}
\usepackage{graphicx}
\usepackage{amssymb}
\usepackage{tabularx}
\usepackage{setspace}
\usepackage{verbatim}
\usepackage{epsfig}
\usepackage[labelfont=bf]{caption}
\usepackage{enumerate}
\usepackage{subfig}
\usepackage{epstopdf}

\newtheorem{lemma}{Lemma}
\newtheorem{proposition}{Proposition}
\newtheorem{corollary}{Corollary}
\newtheorem{fact}{Fact}

\newtheorem{remark}{Remark}

\newtheorem{assumption}{Assumption}

\def\begcen{\begin{center}}
\def\endcen{\end{center}}

\newcommand{\bfp}{\mbox{$p$}}

\newcommand{\bfy}{\mbox{$y$}}

\newcommand{\col}{ \mbox{col} }

\def\calp{{\bfp}}

\def\bfy{{\bf y}}

\def\bfp{{\bf p}}

\def\liminf{\lim_{t \to \infty}}

\def\L2{{\cal L}_2}
\def\L2e{{\cal L}_{2e}}

\def\rea{\mathbb{R}}

\def\x{{x}}

\def\u{{u}}

\def\beal#1{\begin{align}{#1}\end{align}}
\def\begmat#1{\begin{bmatrix}#1\end{bmatrix}}

\def\begali#1{\begin{align}{#1}\end{align}}
\def\begalis#1{\begin{align*}{#1}\end{align*}}

\def\begequarr{\begin{eqnarray}}
\def\endequarr{\end{eqnarray}}
\def\begequarrs{\begin{eqnarray*}}
\def\endequarrs{\end{eqnarray*}}
\def\begarr{\begin{array}}
\def\endarr{\end{array}}
\def\begequ{\begin{equation}}
\def\endequ{\end{equation}}
\def\lab{\label}
\def\begdes{\begin{description}}
\def\enddes{\end{description}}
\def\begenu{\begin{enumerate}}
\def\begite{\begin{itemize}}
\def\endite{\end{itemize}}
\def\endenu{\end{enumerate}}

\def\lef[{\left[\begin{array}}
\def\rig]{\end{array}\right]}

\def\begcen{\begin{center}}
\def\endcen{\end{center}}
\def\begrem{\begin{remark}\rm}
\def\endrem{\end{remark}}
\def\begassum{\begin{assumption}}
\def\endassum{\end{assumption}}
\def\begassums{\begin{assumption*}}
\def\endassums{\end{assumption*}}
\def\begassu{\begin{ass}}
\def\endassu{\end{ass}}
\def\beglem{\begin{lemma}}
\def\endlem{\end{lemma}}
\def\begcor{\begin{corollary}}
\def\endcor{\end{corollary}}
\def\begfac{\begin{fact}}
\def\endfac{\end{fact}}
\def\TAC{{\it IEEE Trans. Automat. Contr.}}
\def\AUT{{\it Automatica}}

\def\liminf{\lim_{t \to \infty}}

\def\L2e{{\cal L}_{2e}}

\def\rea{\mathbb{R}}
\def\intnum{\mathbb{Z}}

\def\col{\mbox{col}}

\def\TAC{{\it IEEE Trans. Automatic Control}}

\def\AUT{{\it Automatica}}

\def\ba{\begin{array}}
\def\ea{\end{array}}

\def\begsubequ{\begin{subequations}}
\def\endsubequ{\end{subequations}}

\def\bfthe{{\boldsymbol\theta}}

\def\bfe{{\bf e}}
\def\calw{{\cal W}}

\def\bfome{{\boldsymbol{\Omega}}}
\def\bfeps{{\boldsymbol \epsilon}}
\def\bfy{{\bf Y}}
\def\ya{y_{1,2}}
\def\yb{y_{3,4}}
\def\yc{y_{5,6}}
\def\za{z_{1,2}}
\def\zb{z_{3,4}}

\DeclareMathOperator{\atantwo}{atan2}

\usepackage{cite}
\usepackage{amssymb,amsfonts}
\usepackage{algorithmic}
\usepackage{graphicx}
\usepackage{textcomp}

\title{An Almost Globally Stable Adaptive Phase-Locked Loop for Synchronization of a Grid-Connected Voltage Source Converter}
\author{Daniele Zonetti, Alexey Bobtsov, Romeo Ortega, Nikolay Nikolaev, Oriol Gomis-Bellmunt\thanks{This work was supported by the Ministry of
Science and Higher Education of Russian Federation, passport of goszadanie no. 2019-0898 and by FEDER/Ministerio de Ciencia, Innovaci\'{o}n y Universidades-Agencia Estatal de Investigaci\'{o}n, Project RTI2018-095429-B-I00. The work of Oriol Gomis-Bellmunt is supported by the ICREA Academia program.}
\thanks{D. Zonetti and O. Gomis-Bellmunt are with the Centre  d'Innovaci\'{o} Tecnol\`{o}gica en Convertidors Est\`{a}tics i Accionaments, Departament d'Enginyeria El\`{e}ctrica, Universitat Polit\`{e}cnica de Catalunya, Barcelona 08028, Spain. (e-mail: daniele.zonetti(oriol.gomis)@upc.edu). }
\thanks{A. Bobtsov is with the Department of Control Systems and Robotics, ITMO University, Kronverkskiy av. 49, Saint-Petersburg, 197101, Russia and with the Laboratory ``Control of Complex Systems'', Institute of Problems of Mechanical Engineering, V.O., Bolshoj pr., 61, St. Petersburg, 199178, Russia (e-mail: bobtsov@mail.ru).}
\thanks{R. Ortega is with the Departamento Acad\'{e}mico de Sistemas
Digitales, ITAM, Rio Hondo 1, Col. Progreso Tizapan, 01080 Ciudad de
M\'{e}xico, Mexico (email: romeo.ortega@itam.mx).}
\thanks{N. Nikolaev is with the Department of Control Systems and Robotics, ITMO University, Kronverkskiy av. 49, Saint-Petersburg, 197101,
Russia (e-mail: nikona@yandex.ru).}}
\begin{document}

\maketitle
\thispagestyle{empty}
\begin{abstract}
In this paper we are interested in the problem of {\em adaptive synchronization} of a  voltage source converter with a possibly weak grid with unknown angle and frequency, but knowledge of its parameters. To guarantee a suitable synchronization with the angle of the three-phase grid voltage we design an adaptive observer for such a signal requiring measurements only at the point of common coupling. Then we propose two alternative certainty-equivalent, adaptive phase-locked loops that ensure the angle estimation error goes to zero for almost all initial conditions. Although well-known, for the sake of completeness, we also present a PI controller with feedforward action that ensures the converter currents converge to an arbitrary desired value. Relevance of the theoretical results and their robustness to variation of the grid parameters are thoroughly discussed and validated in the challenging scenario of a converter connected to a grid with low short-circuit-ratio.
\end{abstract}

\section{Introduction}
\lab{sec1}
%
We have witnessed in the last years a widespread penetration of renewable energy sources into the existing power grid. The integration of these sources is enabled by voltage source converters (VSCs), for which many questions and challenges arise about their operation and control~\cite{BOS}. Suitable operation of a grid-connected VSC is usually enforced by hierarchical controllers that are operated at different time-scales. These controllers structurally rely on a suitably defined rotating frame, whose reference angle implicitly determines the mode of operation of the VSC. Whenever the reference angle is selected to synchronize to, i.e. \textit{follow}, the grid angle, the VSC is said to operate in \textit{grid-following} mode; whenever it is designed, i.e. \textit{formed}, via a suitably defined frequency control loop, the VSC is said to operate in \textit{grid-forming} mode~\cite{MILetal2018}. Therefore, for the operation in grid-following mode the knowledge of the grid angle is essential. Since this information is not available to the system operator, an appropriate algorithm that estimates this signal must be designed. 

{A conventional solution to this problem consists in the design of phase-locked loop (PLL) algorithms~\cite{CHU,TEObook}.  A PLL is a basic, well-known, nonlinear feedback control system  that allows to generate an output signal with an angle that is locked to the angle of a given, input reference signal~\cite{ABR}. This mechanism is typically employed in power systems to lock the angle of the overall VSC control scheme to a suitable reference angle. For more information about PLL architecures from a control theoretical perspective, the interested reader is referred to~\cite{ABR}. Investigation of the stability properties of PLLs, in a variety of perturbed scenarios, has considerably attracted the interest of the power systems and control community. It is often argued that PLL algorithms can be effectively operated at a time-scale much slower than the time-scale at which the VSC controllers operate and physical dynamics evolve, this time-scale separation argument is introduced mainly to simplify the analysis with no further validation for it. Fundamental stability and performance properties of the conventional synchronous reference frame (SRF) and ATAN-PLL, under this widely debated assumption, have been already studied in~\cite{RAN} and more recently in \cite{ZONetal}.  However, many questions arise about the ability of the VSCs to synchronize under especially degraded conditions resulting, for example, from the connection of the VSC to grids characterized by a low short-circuit-ratio \cite{PAPetal}, where the  time-scales of the grid, the VSC and the PLL become specious. In particular, it has been observed that in such scenarios, where the grid is referred as \textit{weak}, instability may be triggered by inappropriate tuning of the PLL gains~\cite{ZHOU} and therefore a time-expensive tuning procedure to tune the latter must be adopted to avoid loss of synchronism and to guarantee acceptable performances~\cite{EGEetal}. As alternative approaches to preserve synchronization, a variety of \textit{ad hoc} outer-loops have been proposed in the literature and their stability properties have been analyzed and illustrated through extensive simulations---see~\cite{WANGetal} and references therein for an excellent overview. However, the stability analysis is carried out based on small-signal approximation of the system's dynamics, therefore providing limited information about the ability of the VSC to maintain synchronization in presence of large perturbations. On the other hand, large-signal stability analysis of PLL-based VSCs has been recently addressed in several papers, see for example~\cite{MANetal,ZHAOetal}. Unfortunately, the analysis builds upon the disputable argument of substantial similarities between PLL and conventional synchronous generator dynamics, often leading to conservative stability certificates whose applicability is highly questionable.

In this paper we provide an {\em almost\footnote{The qualifier ``almost" stands for the fact that convergence is ensured for all initial conditions starting outisde a set that has zero Lebesgue measure and is nowhere dense---so, for all practical purposes, the qualifier is skippable.} globally stable} solution for the synchronization of a grid-connected VSC, upgrading the standard PLL design based on a $\tt dq$ transformation of the system dynamics. More precisely, we complement the conventional synchronization scheme with an {\em observer} that reconstructs the $\tt dq $ grid voltage---an information that is next provided to a PLL to ensure synchronization.}  In a series of papers, starting in \cite{ORTetalscl15} and further generalized in \cite{ORTetalaut21}, a new procedure to design state observers for state-affine systems, called generalized parameter estimation-based observers (GPEBO), has been pursued. The main novelty of GPEBO is that the state observation problem is reformulated as a problem of {\em parameter estimation}---a transformation that is made possible exploiting  the properties of the {\em principal matrix solution} of an unforced linear time-varying (LTV) system $\dot x = A(t)x$ that is constructed in GPEBO. This novel technique has been successfully applied to solve various open problems in observation theory and in several practical applications. 

To tackle the problem at hand we propose in this paper to use the GPEBO technique to derive a linear regressor equation (LRE) needed for the estimation of the the grid voltage. To carry out this task we use a classical least-squares estimator with forgetting factor (LS+FF) \cite{LJUbook,SASBODbook,TAObook}. The addition of the latter was done to preserve the alertness of the estimator and be able to track slowly time-varying parameter variations. In the paper  we show that the parameters of the regressor equation can be---globally and exponentially---estimated under the  classical {\em persistency of excitation} assumption, which turns out to be verified in the present application.

The remainder of the paper is organized as follows. In Section \ref{sec2} we present the model of the system and the synchronization problem formulation. In Section \ref{sec3} we derive the LRE used in the LS+FF estimator. The main result is given in Section   \ref{sec4}. Some illustrative simulation results are presented in Section \ref{sec5}. We wrap-up the paper with concluding remarks and future research in Section \ref{sec6}. In Appendix A we do some additional analysis of the system model.\\

\noindent {\bf Notation.}  Given $n \in \intnum_{+}, q \in \intnum_{+}$, $\mathbb{I}_n$ is the $n \times n$ identity matrix and ${\bf 0}_{n\times q}$  is an $n \times q$ matrix of zeros. $\bfe_q \in \rea^n$ denotes the $q$-th vector of the $n$-dimensional Euclidean basis. For $x \in \rea^n$, we denote the square of the Euclidean norm as $|x|^2:=x^\top x$. Given an $n$-dimensional vector $x=\col(x_1,x_2,\dots,x_n)$ and two integer numbers $r$ and $s$ such that $n \geq r>s \geq 1$ we define the subvector $x_{s,r}:=\col(x_s,x_{s+1},\dots,x_r)$. Given $J:=\begmat{0&-1\\1&0}$ and $\alpha\in \rea$ we denote the rotation matrix as
$e^{J\alpha}=\begin{bmatrix}
       \cos(\alpha)&-\sin(\alpha)\\
       \sin(\alpha)&\cos(\alpha)
       \end{bmatrix}.$ 
The symbol  $\bfeps_t$ stands for a generic signal exponentially converging to zero. Given a differentiable signal $u(t)$, we define the derivative operator $\calp^i[u]=:{d^iu(t)\over dt^i}$ and denote the action of a linear time-invariant (LTI) filter $F(\calp) \in \rea(\calp)$  as $F[u]$. To avoid cluttering the notation we omit the subindex $(\cdot)_{\tt abc}$ usually added to denote symmetric AC three-phase signals represented in the classical ${\tt abc}$ reference frame \cite[Section 2]{SCHetal}.  
%
 \section{Model of the System and Synchronization Objective}
\lab{sec2}
%
In this section we give a brief derivation of the mathematical model of the system we consider in the paper. The interested reader is referred to the classical textbook \cite{YAZIRAbook} or the recent survey paper on the topic \cite{SCHetal} for further details on microgrid systems modeling. 
%
\subsection{Model of the grid connected VSC}
\lab{subsec21}
{
We consider a three-phase, balanced two-level VSC interfaced to a balanced AC grid. The VSC is characterized by six switches---that are supposed to be ideal (neglecting the diodes nonlinear dynamics), bidirectional and mutually synchronized---and is interfaced to the AC subsystem through a phase reactor.

The average model of the three-phase VSC, in the standard {\tt abc} reference frame is given by
\begin{equation}\label{vsc_abc}
L{di \over dt}=-r i+v_\mathrm{dc}m-v,
\end{equation}
where: $i(t)\in\mathbb{R}^3$ denotes the three--phase AC current through the phase reactor; $v(t)\in\mathbb{R}^3$ denotes the three-phase AC voltage at the point of common coupling (PCC) (or equivalently the voltage of the filter); $v_\mathrm{dc}\in\mathbb{R}_{>0}$ denotes the DC voltage; $m\in[\underline{m}\;\overline{m}]^3\subset\mathbb{R}^3$ denotes the three-phase modulation indices;  $L\in\mathbb{R}_{>0}$, $r\in\mathbb{R}_{>0}$ denote respectively the inductance and resistance of the phase reactor. The AC subsystem is modeled using the Thevenin equivalent circuit---which consists of the series connection of a three-phase voltage source with an \textit{RL} circuit---and is interfaced to the VSC via an \textit{RC} filter. Hence, it can be described by
  \begin{equation}\label{grid_abc}
  \begin{aligned}
  L_g{d i_g \over dt}&=-r_gi_g+v -v_g\\
  C\dot{v} &= -i_g+i,
  \end{aligned}
  \end{equation}
where: $i_g(t)\in\mathbb{R}^3$, $v_g(t) \in\mathbb{R}^3$ denote the three--phase ac current and input voltage of the grid, respectively; $L_g\in\mathbb{R}_{>0}$, $r_g\in\mathbb{R}_{>0}$ and $C\in\mathbb{R}_{>0}$ denote respectively the inductance, resistance associated to the grid impedance and the capacitance of the filter. \\

The following assumptions are made in the paper.\smallbreak

\begin{assumption}[AC grid]\label{ass:VAC}  The Thevenin equivalent voltage source is described by a three-phase balanced purely sinusoidal signal:
\begin{equation}
\lab{vg}
v_{g}(t) = {\sqrt{\frac{2}{ 3}}}V_g \begin{bmatrix}
\sin(\omega t)\\
\sin(\omega t-\frac{2}{3}\pi)\\
\sin(\omega t+\frac{2}{3}\pi)
\end{bmatrix}, 
\end{equation}
 where $V_g\in\mathbb{R}_{>0}$ and $\omega\in\mathbb{R}_{>0}$  are {\em both unknown.}
\end{assumption}\smallbreak

\begin{assumption}[DC grid]\label{ass:VDC}  $v_\mathrm{dc}=V_\mathrm{dc}$, with  {\em known} constant $V_\mathrm{dc}\in\mathbb{R}_{>0}$
\end{assumption}\smallbreak

\begin{assumption}[Parameters]
\lab{ass3}
The  positive grid parameters $L_g$ and $r_g$ are {\em known}. 
\end{assumption}

\begin{assumption}[Measurements]
\lab{ass4}
The signals $i_{g}$, $v$ and $i$ are {\em measurable}.
\end{assumption} 

Assumption \ref{ass:VAC} is justified for AC grids that are characterized by a sufficiently large number of synchronous rotating machines, providing large inertia, and/or whenever a sufficiently tight regulation of the frequency is implemented via grid-forming controllers. Assumption \ref{ass:VDC} is legitimized whenever additional controllers take responsibility of ensuring regulation of the DC voltage. Assumption \ref{ass3} is, on the other hand, quite restrictive.  Notice, however, that $V_g$ is not assumed to be known, Aee also point {\bf P4} in Subsection \ref{subsec24}. Assumption \ref{ass4} is quite natural, as it is verified in all practical scenarios.
%
\subsection{A  framework to formulate the synchronization problem}
\lab{subsec22}
For a correct and safe operation of the grid-connected VSC, an essential requirement is that, after reasonable transients, the state and input variables  $\{i_g,v,i,m\}$ converge to a balanced, three-phase AC signal with the common grid frequency $\omega$, but different amplitudes and phase shifts.  This fact complicates both the steady-state analysis and the control design since the steady-states of interest are time-varying. Furthermore, a suitable control of the reactive power is of capital importance for an appropriate operation of the system, and the absence of a clear definition in such coordinates stymies a convenient formulation of the control objectives. 

One way to deal with these issues is to represent the system in a suitably defined $\tt dq$ reference frame that guarantees that the steady-states of interest, in the new coordinates, correspond to constant quantities---see \cite[Section 2]{SCHetal}. However, to guarantee that these signals are constant, the transformation angle---say, $\alpha$---of the $\tt dq $ transformation must verify $\dot\alpha=\omega$ at steady-state, but the latter is unknown. The derivations below will provide an answer to this frequency synchronization problem.

We propose then to apply the $\tt dq$ transformation with an angle $\vartheta-\frac{\pi}{2}$ to all variables of the system \eqref{vsc_abc}, \eqref{grid_abc}, where $\vartheta(t) \in \mathbb{R}$ is the solution of the differential equation
 \begequ
 \lab{dotvarthe}
 \dot \vartheta=u_{\vartheta},
 \endequ
and $ u_{\vartheta}(t) \in \mathbb{R}$ is some function  {\em to be designed}. That is, we define the new signals
$$
(\cdot)_{\tt dq}:=T_{\tt dq}\left(\vartheta - {\pi \over 2}\right)(\cdot)_{\tt abc},
$$
with the definition of $T_{\tt dq}$ of \cite[eq. (2.3)]{SCHetal}. This manipulation yields the system dynamics, 
\begin{align}
    \begin{bmatrix}
    L_g{d{{i}_{g\tt{dq}}} \over dt}\\
    C{d{v}_{\tt{dq}}\over dt}\\
    L{d{i}_{\tt{dq}}\over dt}
    \end{bmatrix}=&
    \begin{bmatrix}
     -r_g\mathbb{I}_2&\mathbb I_2&{\bf 0}_{2\times 2}\\
   -\mathbb I_2& {\bf 0}_{2\times 2} &\mathbb{I}_2\\
    {\bf 0}_{2\times 2}&-\mathbb{I}_2& -r\mathbb{I}_2
    \end{bmatrix}
    \begin{bmatrix}
    i_{g\tt{dq}}\\ v_{\tt{dq}}\\ {i}_{\tt{dq}}
    \end{bmatrix}+\begin{bmatrix}
    L_gJi_{g\tt{dq}} & {\bf 0}_{2\times 2}\\
    CJv_{\tt{dq}} & {\bf 0}_{2\times 2} \\
    {LJi_{\tt{dq}}} & \mathbb{I}_2
    \end{bmatrix} \begmat{u_{\vartheta} \\ u_{\tt{dq}}}
    +\begin{bmatrix}
    - v_{g\tt dq} \\{\bf 0}_{2 \times 1}\\{\bf 0}_{2 \times 1}
    \end{bmatrix}\label{eq:sys-dq}
\end{align}
 where we  defined the new control signal
\begin{equation}\label{eq:vdc-vgdq}
u_{\tt{dq}}  :=V_\mathrm{dc}m_{\tt{dq}},   
\end{equation}
and the rotated grid voltage
\begequ
v_{g\tt dq}:=V_g e^{J\delta}\mathbf{e}_1,
\label{eq:vgdq}
\endequ
where $\delta\in\mathbb R$ is defined as the {\em error signal}
 \begin{equation}
  \delta:=\vartheta-\omega t.
 \end{equation}
This signal clearly satisfies the dynamics equation
\begin{equation}\label{eq:dot-delta}
    \dot\delta=-\omega+u_{\vartheta},
\end{equation}}
and it is treated as a new {\em unmeasurable state variable} added to the system dynamics \eqref{eq:sys-dq}.

Given this scenario the synchronization objective of interest that should be achieved via the suitable selection of the ``control signal" $ u_{\vartheta}$ is to ensure that
{
\begequ
\lab{synobj0}
\liminf \delta(t)={\phi^\mathrm{ref}}\quad \Leftrightarrow \quad \liminf |\vartheta(t) -\omega t|={\phi^\mathrm{ref}},
\endequ
for some desired ${\phi}^\mathrm{ref}\in\mathbb R$, that is, the $\tt dq$ transformation is synchronized with the angle $\vartheta^\mathrm{ref}(t):=\omega t+{\phi}^\mathrm{ref}$. {The question of how to select an appropriate $\phi^\mathrm{ref}$ to guarantee a desired operation of the VSC is briefly discussed in Section~\ref{sec5}.
}
%
\subsection{Observer formulation of the adaptive synchronization problem}
\lab{subsec23}
{
As indicated in the Introduction the adaptive synchronization problem is solved in this paper translating it into an adaptive state observation problem, to which we apply the GPEBO technique. Towards this end, we observe that the signal $v_{g\tt dq}$, that we treat as an {\em unmeasurable state}, satisfies the differential equation 
\begin{equation}\label{eq:dot-vg}
\begin{aligned}
    \dot v_{g\tt dq}&=-\dot\delta J v_{g\tt dq}\\
                    &=-\omega J v_{g\tt dq}+u_\vartheta J v_{g\tt dq}.
    \end{aligned}
\end{equation}
Therefore, the synchronization objective \eqref{synobj0} can be recast as follows.
\begenu
 \item[{\bf SO}] Select the ``control signal" $ u_{\vartheta}$ to ensure that the unmeasurable state $v_{g\tt dq}$ satisfies
\begequ
\lab{synobj}
\liminf v_{g\tt dq}(t)=V_ge^{J{\phi}^\mathrm{ref}}\mathbf{e}_1.
\endequ
\endenu
}
Clearly, to solve this problem it is necessary to design an adaptive observer for the state $v_{g\tt dq}$, which is the main contribution of this paper. 

To enhance readability, before proceeding to formulate the adaptive synchronization problem, we find convenient to rewrite the system   \eqref{eq:sys-dq}, \eqref{eq:dot-vg} using the standard control theory notation. Towards this end define the {(non-measurable) state, and (measurable)} output\footnote{The fact that $i_{g\tt{dq}},v_{\tt{dq}}$ and ${i}_{\tt{dq}}$ are measurable stems from Assumption \ref{ass4} and the fact that the rotation angle $\vartheta$, defined in \eqref{dotvarthe}, is known.} and input vectors as
$$
{x:={1 \over L_g}v_{g\tt dq}\in\mathbb R^2,}\;
y:= \begmat{
    i_{g\tt{dq}}\\ v_{\tt{dq}}\\ {i}_{\tt{dq}}} \in \rea^6,\;u:= \begmat{u_{\vartheta} \\ u_{\tt{dq}}} \in \rea^3,
$$
respectively, and write the system dynamics as 
\begali{
\dot y =&    
    \begmat{
     -{r_g \over L_g}\mathbb{I}_2&{1 \over L_g}\mathbb I_2&{\bf 0}_{2\times 2}\\
   -{1 \over C}\mathbb I_2& {\bf 0}_{2\times 2} &{1 \over C}\mathbb{I}_2\\
    {\bf 0}_{2\times 2}&-{1 \over L}\mathbb{I}_2& -{r \over L}\mathbb{I}_2
    }y+\begmat{
    J \ya & {\bf 0}_{2\times 2}\\
    J \yb & {\bf 0}_{2\times 2} \\
    J \yc & {1 \over L}\mathbb{I}_2
    }u
    +\begmat{
    -x\\{\bf 0}_{2 \times 1}\\{\bf 0}_{2 \times 1}
    },
 \label{sys}
}
where the state $x$ satisfies the differential equation
\begequ
\lab{sysx}
{\dot x}  {=-\omega J x+u_1 J x}
\endequ

\noindent {\bf Adaptive synchronization problem} {Consider the system \eqref{vsc_abc}, \eqref{grid_abc} verifying  {\bf Assumptions A1-A4}, and its $\tt dq$ representation \eqref{sys}. Design a globally, exponentially convergent {\em adaptive observer}  of the signal {$x$ of the form
\begali{
	\nonumber
	\dot {\xi} & =f(\xi,y,u)\\
	\lab{adaobs}
	\hat x & =h(\xi,y,u),
}
where $f: \rea^{n_\xi} \times \rea^6 \times \rea^3  \to \rea^{n_\xi}$, with the dimension of the observer state $n_\xi>0$ to be defined, and $h:  \rea^{n_\xi} \times\rea^{6}  \times \rea^3 \to \rea$, and a function $e_\delta:\mathbb R^2\rightarrow [-\pi\;\pi)$ such that the  {\em certainty-equivalent} adaptive PLL
\begali{
\nonumber
\dot x_c &= e_\delta(\hat x) \\
\lab{adapi}
u_{1} &=-K_{P\delta} e_\delta(\hat x) -K_{I\delta} x_c,
} 
with gains $K_{P\delta}>0$, $K_{I\delta}>0$ ensures  the synchronization objective \eqref{synobj}. Moreover, this should be achieved---for \textit{almost} all the system and controller initial conditions---guaranteeing that all signals remain bounded.}
\subsection{Discussion}
\lab{subsec24}
The following remarks are in order.\\

\noindent {\bf P1}
We underscore the fact that $\omega$ is not known and  {$x$} is not available for measurement. Therefore, in view of the dynamics \eqref{sysx}, we are dealing with the task of designing an adaptive observer, where both of them have to be estimated. We recall here that since we are dealing with high-inertia grids it is reasonable to assume that $\omega$ is constant. \\

\noindent {\bf P2}
The function $e_\delta$ is usually referred as a phase detector, that is a nonlinear function whose output contains the phase difference between two input oscillating signals. The most popular phase detectors employed in the power systems literature are based on suitable $\tt dq$ transformations of the input oscillating signals~\cite{WANGetal}.\\

\noindent {\bf P3}
{If we replace the state estimate $\hat x$ in  the phase-locked loop \eqref{adapi} by the actual state $x$, and design alternatively
\begin{equation}\label{eq:PD}
\begin{aligned}
e_{\delta}(x)&=-\cos({\phi}^\mathrm{ref})x_2+\sin({\phi}^\mathrm{ref})x_1,\\
\mathrm{or}\quad
e_{\delta}(x)&=-\atantwo(x_2,x_1)+{\phi}^\mathrm{ref},
\end{aligned}
\end{equation}
 we obtain the nonlinear, second-order model of the generalized SRF-PLL or ATAN-PLL, whose almost global stability properties have been already demonstrated in \cite{RAN,ZONetal}. Conventional SRF- and ATAN-PLL aligned with the $\tt q$ axis can be further recovered by picking ${\phi}^\mathrm{ref}=0$.}\\

\noindent {\bf P4}
An algorithm that includes the {\em estimation of the grid parameters} $r_g$ and $L_g$ can be further designed. As expected, this scheme is more complicated than the previous one and, in particular, gives rise to a nonlinearly parametrized regression equation (NLPRE)---instead of the LRE that we had before. Since simulations have shown that the first solution is quite insensitive to imprecise knowledge of these parameters, we tend to believe that this more complicated solution, although of theoretical interest, is not necessary for a practical application.  
%
\section{Reducing the  Problem of Observation of $x$ to Parameter Estimation}
\lab{sec3}
%
In this section we apply the GPEBO procedure \cite{ORTetalaut21} to translate the problem of observation of the state variable  {$x$} to one of parameter estimation.   Following the standard procedure \cite{LJUbook,SASBODbook}, to carry-out the parameter estimation, we construct  a vector LRE that is used in a classical LS adaptation algorithm.  In this way, the estimated parameters can be used to generate the estimated state $\hat x$ for the implementation of the adaptive PLL \eqref{adapi} with guaranteed stability properties.
\subsection{Construction of the LRE and expression for  {$x$}}
\lab{subsec31}

\begin{proposition}
\lab{pro1}\em
Consider the system \eqref{vsc_abc}, \eqref{grid_abc} verifying  {\bf Assumptions A1-A4}, and its $\tt dq$ representation \eqref{sys}. The unmeasurable state $x$ satisfies the following algebraic equation
\begequ
\lab{xi}
\x=\calw \bfthe,
\endequ
with $\calw (t) \in\rea^{2 \times 3}$ a {\em measurable} signal and $\bfthe\in\rea^3$ a vector of {\em constant, unknown} parameters which satisfies the LRE
\begequ
\lab{lre}
\bfy=\bfome   \bfthe + \bfeps_t,
\endequ 
where the signals  $\bfy(t)\in\rea^2$ and $\bfome  (t) \in\rea^{2 \times 3}$ are {\em measurable}.
\end{proposition}

 \begin{proof}
The measurable signal $\ya$ can be rewritten as:
  \begali{
\lab{dotya} 
     \dot y_{1,2} &=q- {x},
}
 where we defined the signal
 \begali{
\nonumber
  q & :=-\frac{r_g}{L_g}\ya + \frac{1}{L_g}\yb+J \ya u_1,
 }
and underscore the fact that it is {\em measurable} and recall that
\begalis{
{\dot x}  &=-\omega J x+u_1 J x\\
 &=u_1 J x-\omega J(q-\dot y_{1,2}).
}
Following the GPEBO construction we define the following dynamic extension with states $z \in\rea^4$ 
\begalis{
	\dot z_{1,2}& =   u_1 J \za + J q\\
\dot z_{3,4}& =   u_1 J( \zb-J \ya).
}
Define now the signal
\beal{
	\lab{e}
	e:= \omega (\za + \zb - J \ya)  +  {x}.
}
whose time derivative is given as
\begalis{
\dot e=& \omega [ u_1 J \za + J q+ u_1 J (\zb -  J\ya) - J    \dot y_{1,2}]  +  {\dot x} \\
 = & \omega u_1 J  (\za +  \zb - J \ya )  + u_1 J  {x} \\
=&  u_1 J [ \omega (\za + \zb - J \ya) +  {x}] \\
=&u_1J e.
}
Following the GPEBO method we introduce the matrix equation
$$
\dot \Phi=u_1J\Phi,\;\Phi(0)= \mathbb{I}_2,
$$
that is the principal matrix solution of the LTV system $\dot e=-u_1Je$.\footnote{The choice of initial conditions $\Phi(0)= \mathbb{I}_2$ is done, without loss of generality, to simplify the notation---see \cite[Property 4.4]{RUGbook} for the general case.} Invoking the properties of this matrix \cite{RUGbook}, we can write
$$
e(t)=\Phi(t)e(0),\;\forall t \geq 0..
$$
Replacing the identity above in \eqref{e} yields
\begequ
\lab{xi1}
	 {x}= \Phi e(0)-\omega (\za + \zb - J \ya),
\endequ
To establish \eqref{xi} we introduce  the definitions
\begalis{
 {\calw } &:= \begmat{-(\za + \zb - J \ya) & \Phi}\\
\bfthe &:=\begmat{\omega \\ e(0)}.
}	

To obtain the LRE \eqref{lre} we replace \eqref{xi1} in \eqref{dotya} to get
\begequ
\lab{firlre}
     \dot y_{1,2} =q+\omega (\za + \zb - J \ya)- \Phi e(0).
\endequ

To generate the LRE needed for the parameter estimation we proceed, as usual, applying some filtering to the equation above. Towards this end, we fix an LTI, stable filter
	\begequ
	\lab{f}
	F(\calp)={\lambda \over \calp+\lambda},
	\endequ
	with $\lambda>0$. We apply this filter to \eqref{firlre}
	\begalis{
		\calp F(\calp)[\ya]  = &F(\calp)[q]+F(\calp)[\za + \zb - J \ya]\omega-F(\calp)[\Phi]  e(0),
		}
which may be written as	the LRE \eqref{lre} with the definitions
\begalis{
\bfy &:= \calp F(\calp)[\ya] - F(\calp)[q]\\
\bfome   &:= -F(\calp)[ {\calw }],
}
and $\bfeps_t$ stands for the effect of the initial conditions of the various stable filters.	
\end{proof}
\subsection{Summary of the adaptive observer of the state $x$}
\lab{subsec32}
%
The adaptive observer given in \eqref{adaobs} is summarized, with $ {{n_\xi}=8}$, as follows:
\begali{
\lab{imp}\dot\xi=
	  \begmat{\dot z_{1,2}\\ \dot z_{3,4}\\\dot \Phi_1\\\dot \Phi_2}= \begmat{u_1Jz_{1,2}-\frac{r_g}{L_g}{J}\ya + \frac{1}{L_g}{J}\yb- \ya u_1 \\ u_1J(z_{3,4}-J y_{1,2})\\ u_1J \Phi_1\\u_1J\Phi_2 \\ },
}
with
$$
\xi(0)=\col{(z_{1,2}(0),z_{3,4}(0),\bfe_1,\bfe_2)}
$$
where \footnote{Please note that we have introduced the notation $\Phi=:\begmat{\Phi_1&|&\Phi_2}$ where $\Phi_i(t) \in \rea^2,\;i=1,2$, are the columns of the $2 \times 2$ matrix $\Phi$. Also, to avoid cluttering the notation the readout map $h({\xi},y,u)$ of  \eqref{adaobs} is replaced by the expression \eqref{hatxi} with the explicit appearance of the estimated parameters $\hat \bfthe$.}
\begequ
\lab{hatxi}
\hat x =\begmat{(-\za - \zb + J \ya) &  \Phi_1 &  \Phi_2}\hat \bfthe,
\endequ
and $\hat \bfthe$ is generated by the parameter estimator given in Section \ref{sec4}. 
 \section{Main Result}
\lab{sec4}
%
In this section we present our main result. As indicated in the Introduction besides the synchronization objective to be accomplished with a suitable selection of the signal $u_1$, which is our main concern in the paper, there is also the requirement to design $u_{2,3}$ to drive the grid current   $\yc$ towards a desired value  $\yc^\mathrm{ref}$. This task is pretty standard and has a rather trivial solution. 

To enhance readability we split into two different subsections the solutions to both problems.  
\subsection{A PI controller to regulate the VSC current to a constant desired value}
\lab{subsec41}
Before proceeding to present the solution of the problem we make the following observations. First, we recall that, as shown in Appendix \ref{appa}, the set of assignable equilibria of the system \eqref{sys} can be 
characterized by a mapping
$$
y_{1,4}=W( {x},y_{5,6}),
$$
{with $x=V_ge^{J\delta}\mathbf{e}_1$.} Hence,  ${\delta}$ and $y_{5,6}$ can be fixed at {\em any constant desired} value ${\phi^\mathrm{ref}}$ and $y^\mathrm{ref}_{5,6}$. 

Second, the dynamics of the VSC current may be written in the form
$$
\dot y_{5,6}=-{r \over L}y_{5,6} + {1 \over L} u_{2,3}-[{1 \over L} y_{3,4} -Jy_{5,6}u_1].
$$
Hence, if the control signal contains a term that cancels the term in brackets the remaining dynamics is linear. This, partial linearizing approach is the one adopted below.
  
\begin{proposition}
\lab{pro2}\em
Consider the system \eqref{vsc_abc}, \eqref{grid_abc} verifying  {\bf Assumptions A1-A4}, and its $\tt dq$ representation \eqref{sys}.  Define the error signal
$$
\tilde y_{5,6}:=\yc-\yc^\mathrm{ref}.
$$ 
The classical partial linearizing plus PI controller
\begin{equation}
\label{eq:current-controller}
\begin{aligned}
\dot x_{c} &= \tilde y_{5,6}\\
u_{2,3} &=-K_{P}\tilde y_{5,6} -K_{I}x_{c}+y_{3,4}-LJ y_{5,6} u_1,
\end{aligned}    
\end{equation}
with the PI gains $K_{P} \in \rea^{2 \times 2}$ and  $K_{I} \in \rea^{2 \times 2}$ positive definite, ensures that 
\begalis{
\liminf \tilde y_{5,6}(t)& =0,\quad (exp).
}
\end{proposition}

\begin{proof}
In view of the cancellation of the additive terms $-y_{3,4}+LJ y_{5,6} u_1$ the closed-loop dynamics of $\tilde y_{5,6}$ is the simple LTI system
$$
(\calp^2+ K_{P} \calp+K_{I})[\tilde y_{5,6}]=0,
$$
whose stability for positive gain $K_P$ and $K_I$ is trivially established. 
 \end{proof}
%
\subsection{An almost globally stable adaptive {PLL}}
\lab{subsec42}
%
In this section we complete the design of the proposed adaptive PLL. Towards this end, we propose to use a standard LS+FF scheme to estimate the parameters $\bfthe$ using the LRE \eqref{lre}. Then, we replace---in the classical certainty-equivalent way---these estimated parameters in the full information controller   \eqref{xi} to obtain  \eqref{hatxi}, and  then use the latter in the adaptive  {PLL} \eqref{adapi}. 

To state the main result the following classical assumption of persistent excitation of the regressor matrix $\bfome$ is imposed. As is well-known \cite{SASBODbook,TAObook} this assumption is necessary and sufficient to guarantee the global exponential convergence of the estimated parameters.\footnote{See the comments at the end of Section \ref{sec5} for a relaxation of this assumption.} 

\begin{assumption}[Persistent Excitation]\label{ass5}  
The regressor matrix $\bfome$ of the LRE \eqref{lre} is persistently exciting. That is, there exists constants $C_c>0$ and $T>0$ such that
	\begalis{
		&\int_{t}^{t+T} \bfome  ^\top(s) \bfome  (s)  ds \ge C_c I_3,\;\forall t \geq 0.
	}
\end{assumption}

The main result of the paper is the proof of almost global synchronization, which is established immediately from the global exponential convergence of the estimated parameters and the almost global stability properties of the SRF-PLL and ATAN-PLL established in \cite{RAN} and \cite{ZONetal}, respectively.

	\begin{proposition}
		\label{pro3}\em
		Consider the grid-connected VSC system  \eqref{vsc_abc}, \eqref{grid_abc} verifying  {\bf Assumptions A1-A4}, its $\tt dq$ representation \eqref{sys} and the associated LRE \eqref{lre} constructed as indicated in Proposition \ref{pro1}. Assume $\bfome  $ verifies  {\bf Assumption \ref{ass5}}. Let the adaptive PLL be given by \eqref{adaobs}, \eqref{adapi},   \eqref{imp} and  \eqref{hatxi}, {with $e_\phi(\hat x)$ given by \eqref{eq:PD}}, and with the estimated parameters $\hat \bfthe$  generated via the following LS+FF parameter estimator
\begsubequ
		\lab{intestt1}
		\begali{
			\lab{thegt1}
			\dot{\hat \bfthe }   & =\alpha F   \bfome^\top    (\bfy  -\bfome    \hat\bfthe  ),\; \hat\bfthe(0)=:\theta _{0} \in \rea^3\\
			\lab{phit1}
		\dot {F}  & =\;\left\{
\begarr{cc}
 -\alpha F   \bfome^\top    \bfome     F   {+ \beta F   } & \mbox{if} \; \|F  \| \leq M   \\
 0 & \mbox{otherwise}
\endarr
\right.}
		\endsubequ
		with  $F(0)={1 \over f_0} I_3$ and tuning gains $\alpha>0,\;f_0>0,\; {\beta \geq 0}$ and $M>0$. Then, the synchronization objective \eqref{synobj} is almost globally satisfied with all signals bounded.
			\end{proposition}

%
\section{Simulation Results}
\lab{sec5}
{To validate the theoretical results we consider a VSC with rated power of $1\; GW$ interfaced to a $320\;kV$ transmission grid operating at $\omega=50\; Hz$ and with a SCR of $1$, therefore referred as \textit{weak}.  The parameters of the transmission grid are given as follows:
$$
L_g=0.33\;H,\;R_g=10.24\;\Omega,\;V_g=261\;kV,
$$
whereas the parameters of the filter and of the VSC are
$$
C=5.29\;\mu F,\;L=0.065\;H,\;R=1.02\;\Omega.
$$
We further suppose that the grid is characterized by a relatively high inertia and thus we can safely assume that the frequency $\omega$ remains constant over the time-scale of interest.
To better contextualize the simulation results, let us further recall that---in accordance with standard operation of the VSC---the current references $i_{\tt dq}^\mathrm{ref}$  and phase reference $\phi^\mathrm{ref}$ are established by the power system operator via the power flow equations. The rationale behind this choice is to ensure that a desired active power demand $P^\mathrm{ref}$ and voltage amplitude $V^\mathrm{ref}:=\sqrt{3/2}V_g$ at the PCC are achieved at steady-state. We have indeed the following relations:
\begin{equation}\label{eq:h}
\begin{aligned}
\phi^\mathrm{ref}&=h_\phi(r_g,L_g,V_g,\omega,P^\mathrm{ref},V^\mathrm{ref}),\\
i_{\tt dq}^\mathrm{ref}&=h_{\tt dq}(r_g,L_g,V_g,\omega,P^\mathrm{ref},V^\mathrm{ref}),
\end{aligned}
\end{equation}
for some functions $h_\phi$, $h_{\tt dq}$. Based on these considerations, we identify  two main scenarios for simulations, evaluating the ability of the controlled VSC to maintain synchronization in the face of variations of the $\tt dq$ currents and phase references---stemming from the variation of the power demand---and of the grid parameters. The control objectives to be fulfilled are the following.
\begin{enumerate}
    \item \textit{Nominal conditions.} Preserve synchronization of the VSC to the angle $\vartheta^\mathrm{ref}(t)=\omega t+\phi^\mathrm{ref}$ and ensure regulation of the $\tt dq$ currents to $i_{\tt dq}^\mathrm{ref} $.  Note that the scheme must be effective assuming variation of the active power references over the entire admissible range---a specification that is known to be particularly challenging when the grid has a short-circuit-ratio inferior to $1.5$.
    \item \textit{Perturbed conditions.} Preserve synchronization of the VSC to the angle $\bar \vartheta(t)=\omega t+\bar\phi$, for some $\bar\phi\in\mathbb R$ and ensure regulation of the $\tt dq$ currents to $i_{\tt dq}^\mathrm{ref}$, following abrupt variations of any of the parameters describing the Thevenin equivalent, that is the grid frequency $\omega$, the grid voltage amplitude $V_g$, the grid inductance $L_g$ and the grid resistance $r_g$.
\end{enumerate}

In the simulations we consider the current controller \eqref{eq:current-controller}, together with the adaptive ATAN-PLL\footnote{Similar results were obtained for the adaptive SRF-PLL, but are not reported for the sake of brevity.} given by \eqref{adaobs}, \eqref{adapi},   \eqref{imp} and  \eqref{hatxi}, with $e_\delta(\hat x)=\atantwo(\hat x_2,\hat x_1)$ and where the estimated parameters $\hat \bfthe$ are generated via the LS+FF \eqref{intestt1}---see Fig. \ref{fig:scheme} for a block diagram representation of the proposed scheme. The gains of the ATAN-PLL are set to $K_{P\delta}=200$, $K_{I\delta}=10^3$; the gains of the current controller \eqref{eq:current-controller} are set to $K_P=250$, $K_I=50\cdot 10^3$; the gains of the LS+FF are set to $\alpha=\beta=10^3$, $M=100$ and $f_0=1$.  
\begin{figure*}
\centering
\includegraphics[width=0.6\textwidth]{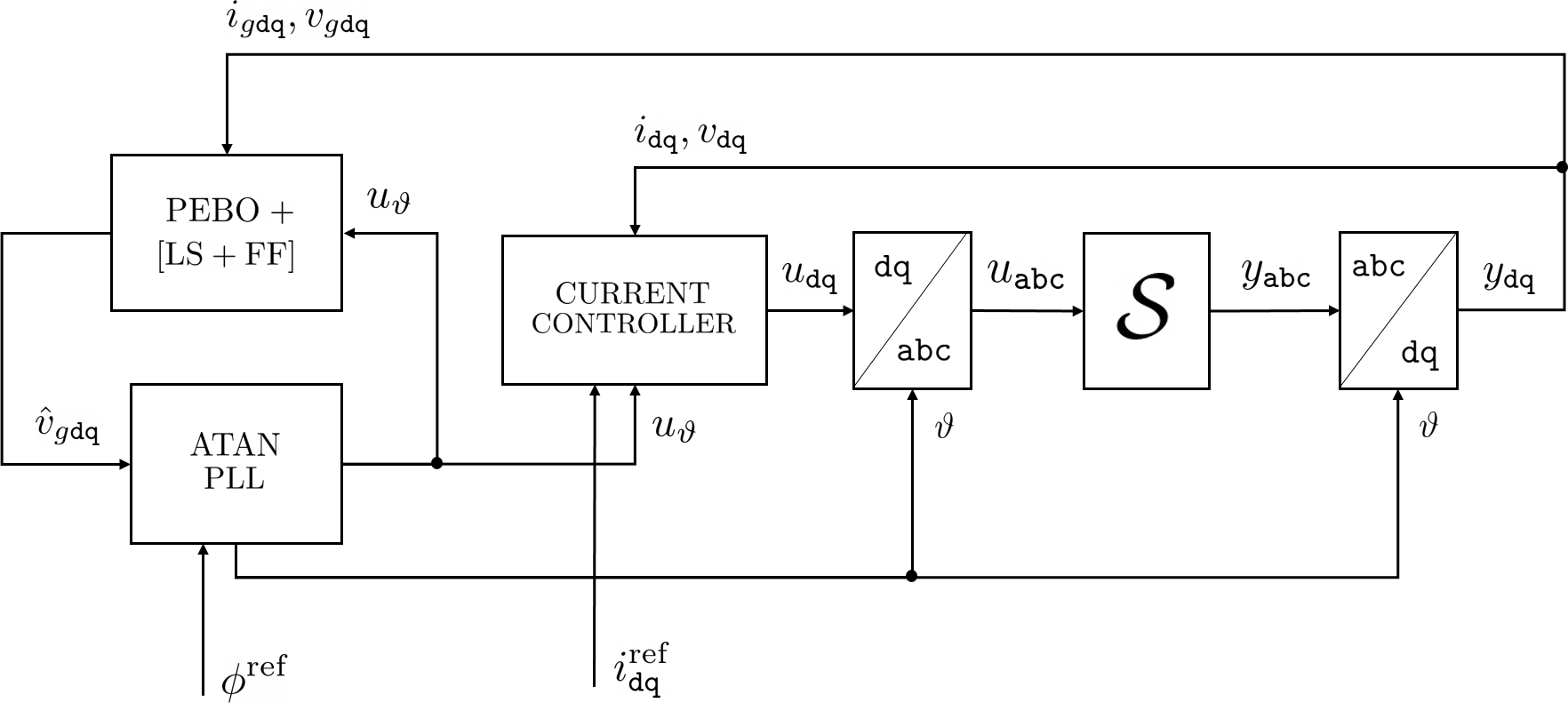}
    \caption{Block diagram representation of the proposed scheme where $\mathcal S$ denotes  \eqref{vsc_abc}-\eqref{vg}.}  \label{fig:scheme}
\end{figure*}

To evaluate the performance of the proposed solution in {\em nominal conditions}, we assume that active power reference varies by step in the time interval $[0\;2]\;s$. Since the system parameters are known, so are the functions $h_\phi$ and $h_{\tt dq}$ and then the corresponding current and phase references can be immediately calculated via \eqref{eq:h}. The $\tt dq$ currents and phase responses are illustrated in Fig.\ref{fig:nominal-idq-phi}. It can be seen that the corresponding references are correctly and fastly tracked over the entire admissible power range. This should be contrasted with the responses of the conventional {\em non-adaptive} ATAN-PLL, which for the same gains becomes unstable for large power demands, as seen in Fig. \ref{fig:nominal-power-detail} where a detail of the phase responses of the adaptive and non-adaptive PLL is given. \\
To evaluate the performance of the adaptive ATAN-PLL in {\em perturbed conditions}, we assume that the VSC is operating at nominal voltage and $75\%$ of its rated power and that it is affected by a disturbance at $1\;s$. First, we model the disturbance as a drop of $30\%$ in the grid voltage $V_g$. As a second scenario, we consider a frequency drop of $1\;Hz$. Finally, a last scenario focuses on the case where, as a result of a trip of the grid impedance, a change in the short-circuit-ratio from $1$ to $0.75$ occurs, corresponding to a simultaneous change of $25\%$ of both the grid inductance and resistance. It must be remarked that, differently from the voltage and frequency drop scenarios, a change in the short-circuit-ratio violates the assumption of Proposition~\ref{pro2}, requiring indeed \textit{a priori} knowledge of the grid impedance. Simulations with respect to this type of disturbance are therefore carried out to assess robustness of the proposed adaptive ATAN-PLL to poor knowledge of such parameters. The estimate of the voltage amplitude, the $\tt{dq}$ currents and the phase in the voltage drop scenario are illustrated in Fig. \ref{fig:voltage-drop}, whereas the estimate of the frequency, the $\tt{dq}$ currents and the phase in the frequency drop scenario are illustrated in Fig. \ref{fig:frequency-drop}.
\begin{figure*}[h]
\centering
\includegraphics[width=0.49\textwidth]{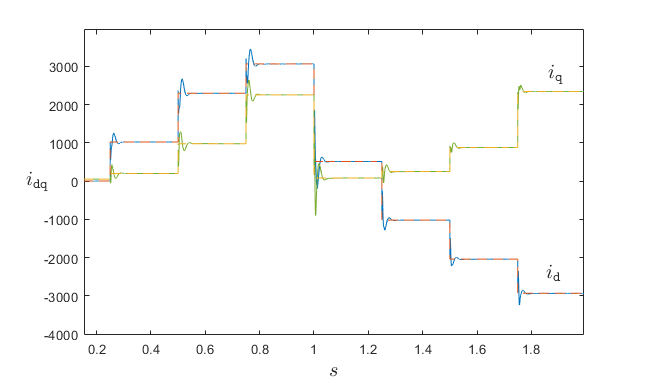}
  \includegraphics[width=0.49 \textwidth]{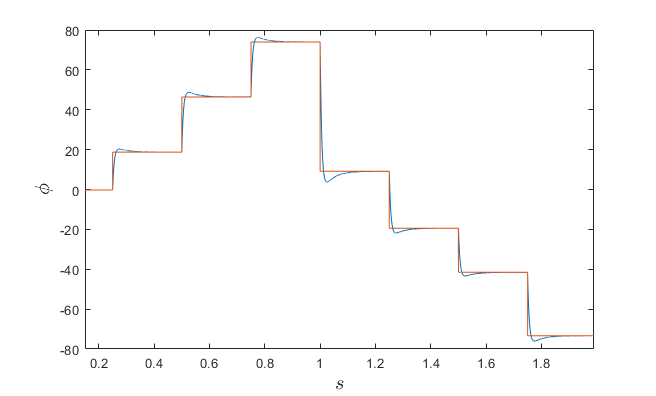}
    \caption{The $\tt dq$ currents and phase responses under the adaptive ATAN-PLL based on the LRE \eqref{lre}, following variation of the demand over the entire admissible range.}  \label{fig:nominal-idq-phi}
\end{figure*}
\begin{figure}[h]
\centering
\includegraphics[width=0.49\textwidth]{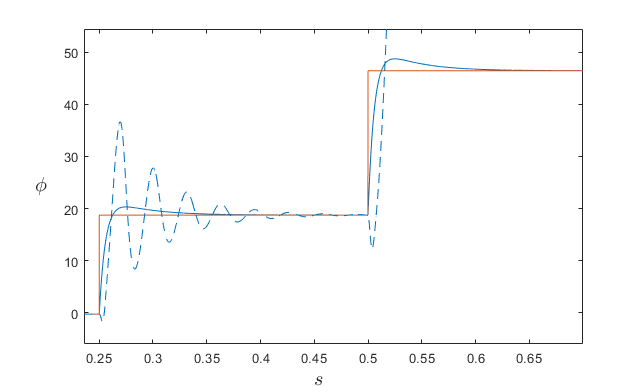}
    \caption{Detail of the phase responses under the adaptive ATAN-PLL based on the LRE \eqref{lre} (solid line) and conventional ATAN-PLL (dotted line), for power demands of $0.4\;GW$ ($\approx 19^\circ$), then $0.9\;GW$ ($\approx 45^\circ$).}  \label{fig:nominal-power-detail}
\end{figure}
\begin{figure*}[h]
\centering
\includegraphics[width=0.32\textwidth]{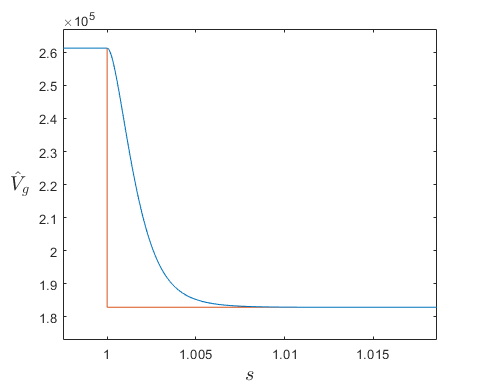}
\includegraphics[width=0.32\textwidth]{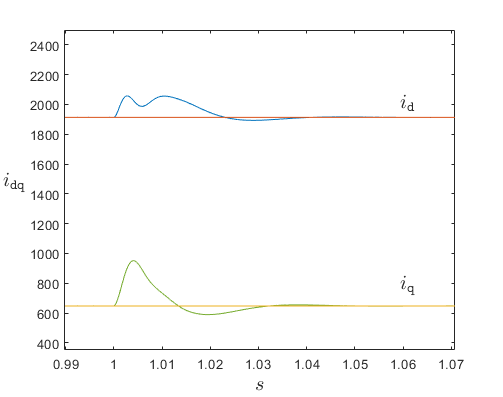}
\includegraphics[width=0.32\textwidth]{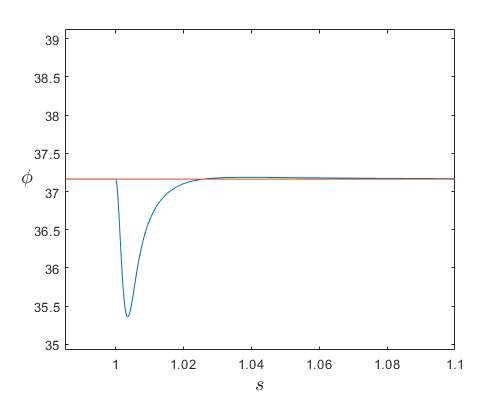}
    \caption{Responses of the voltage amplitude estimate, the $\tt dq$ currents and the phase under the adaptive ATAN-PLL based on the LRE \eqref{lre}, following a drop of $30\%$ of the grid voltage amplitude with respect to its nominal value.}  \label{fig:voltage-drop}
\end{figure*}
\begin{figure*}[h]
\centering
\includegraphics[width=0.32\textwidth]{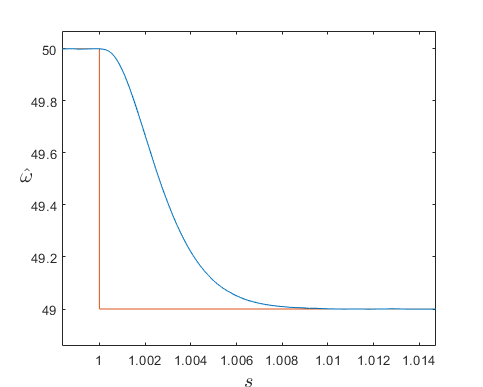}
\includegraphics[width=0.32\textwidth]{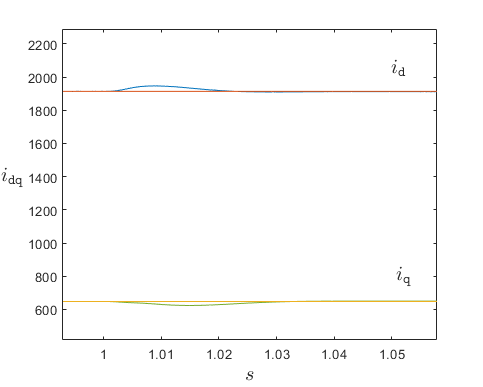}
\includegraphics[width=0.32\textwidth]{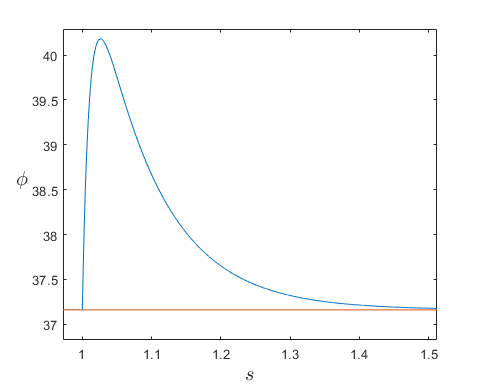}
    \caption{Responses of the frequency estimate, the $\tt dq$ currents and the phase under the adaptive ATAN-PLL based on the LRE \eqref{lre}, following a drop of $1\; Hz$ of the grid frequency.}  \label{fig:frequency-drop}
\end{figure*}
In both situations the estimation of the unknown parameters is achieved smoothly in less than $15\; ms$, further guaranteeing that the $\tt{dq}$ currents and phase are quickly restored to their desired values. In Fig. \ref{fig:SCR-drop} we illustrate the $\tt dq$ currents and phase responses following a trip of the grid impedance.
\begin{figure*}[h!]
\centering
\includegraphics[width=0.49\textwidth]{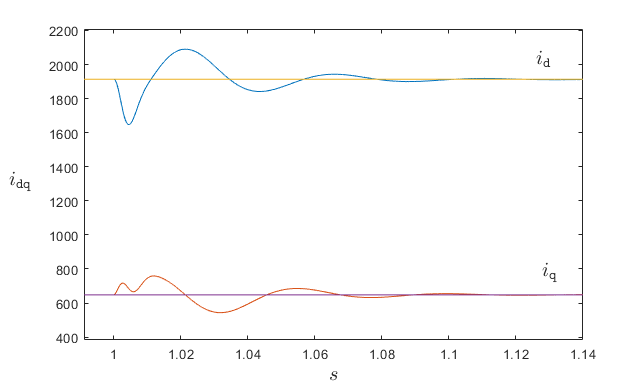}
\includegraphics[width=0.49\textwidth]{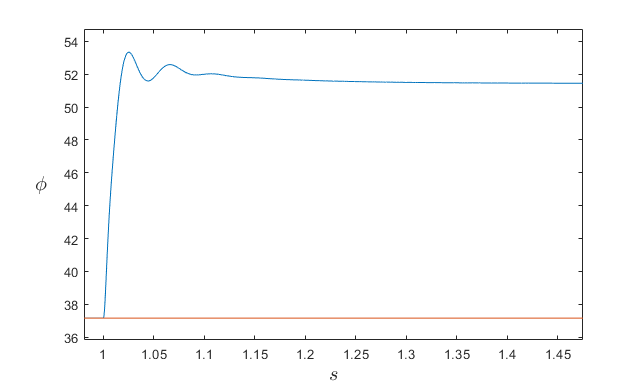}
    \caption{Responses of the $\tt dq$ currents and the phase under the adaptive ATAN-PLL based on the LRE \eqref{lre}, following a trip of the grid impedance reducing the SCR from $1$ to $0.75$.}   \label{fig:SCR-drop}
\end{figure*}
It can be seen that immediately after this event the $\tt dq$ currents are rapidly restored and that the phase remains stable, demonstrating robustness of the proposed scheme face to variation of the grid inductance and resistance.\\
Before closing this section we make the remark that it is actually possible to relax the excitation {\bf Assumption 5}. Indeed, it has been shown in \cite{PYRetal} that global exponential convergence of an LS+FF estimator with dynamic regression extension  is ensured with the {\em strictly weaker} interval excitation condition \cite{KRERIE}.  This scheme was also implemented in the simulations but the performance improvement with respect to the LS+FF was negligible---hence, for brevity, the simulated results are omitted.


\section{Concluding Remarks}
\lab{sec6}
%
We have presented in the paper an adaptive solution to the problem of almost global synchronization of a grid-connected VSC. The scheme consists of an adaptive observer of the grid voltage and a standard PLL. The standing assumptions adopted for the problem are quite standard, with the only critical one being the knowledge of the grid parameters. 
It is clear from the proof of Proposition \ref{pro1} that the key step of deriving the LRE \eqref{lre} is not affected if we add the dynamics of the DC voltage of the converter. It is the authors belief that, in the present context, this additional complication is unnecessary. 

Our current research is oriented towards the {\em experimental implementation} of the proposed synchronization scheme and its comparison with the standard non-adaptive PLL-based strategies. We hope to be able to report these results in the near future.

We are also working on the development of a new adaptive observer design for the case when the grid parameters $r_g$ and $L_g$ are unknown. The prize that is payed with this modification is that the number of the parameters to be estimated will be increased jeopardizing the verification of the excitation requirement of  {\bf Assumption \ref{ass5}}.

%
\appendix
%
\section{Assignable Equilibria}
\lab{appa}
%

To streamline the derivation of the assignable equilibria, we define the following matrices: 
$$
    Z_g:=L_g\omega J-r_g\mathbb{I}_2,\quad Z:=L\omega J-r\mathbb{I}_2,\quad
    Y_\mathrm{c}:=C\omega J,
$$
and the block matrix
$$
Q:=\begmat{ Z_g & -\mathbb{I}_2 \\ \mathbb{I}_2 & -Y_\mathrm{c}}
$$
where  $Z_g \in \rea^{2 \times 2}$, $Z \in \rea^{2 \times 2}$ and $Y_\mathrm{c} \in \rea^{2 \times 2}$ are the grid and converter impedances and the filter admittance, respectively.

The lemma below, whose proof follows trivially setting the left-hand side of \eqref{sys} to zero and recalling \eqref{eq:vgdq}, characterizes the assignable equilibrium set and the corresponding equilibrium control.
 {
\begin{lemma}\em
\lab{lem1}
Consider the system \eqref{sys}, \eqref{sysx}.
\begite
\item[(i)]
    The set of assignable equilibria is given by 
    \begin{equation}
   \mathcal{E}:=\{( {x},y)\in\mathbb R^8\;|\;x=V_ge^{J\delta}\mathbf{e}_1,   y_{1,4}=W(x,y_{5,6}),\delta\in\mathbb R\} 
    \end{equation}
        where the mapping $W:\rea^2 \times \rea^2 \to \rea^4$ is defined as
        $$
                    W(x,y_{5,6}):=Q^{-1}\begmat{x \\  \yc},
        $$
from which we conclude that the desired values for $\delta$ and $y_{5,6}$ can be freely chosen.
\item[(ii)]     
Given an assignable equilibrium $(\bar x,\bar y)\in\mathcal E$, with \mbox{$\bar x=V_ge^{J\bar\phi}$}, $\bar\phi\in\mathbb R$, the corresponding equilibrium control is
        $$
        \bar u=\begmat{\omega\\ \bar y_{3,4}-Z \bar y_{5,6}}.
        $$
\endite
\end{lemma}}

\end{document}